\documentclass[a4paper,11pt]{article}
\usepackage[bbgreekl]{mathbbol}
\usepackage[small]{titlesec}
\usepackage[utf8]{inputenc}
\usepackage[T1]{fontenc}  
\usepackage[backend=biber,style=numeric-comp,natbib=false,giveninits=true,doi=true,isbn=false,url=false,date=year,sorting=none]{biblatex}
\DeclareNameAlias{default}{last-first}

\DeclareFieldFormat[article]{title}{#1}
\renewbibmacro{in:}{}
\DeclareFieldFormat[article]{pages}{#1}
\DeclareFieldFormat{journal}{#1}
\renewcommand\bf\bfseries

\renewbibmacro*{volume+number+eid}{%
	\printfield{volume}%
	\setunit*{\addnbspace}
	\printfield{number}%
	\setunit{\addcomma\space}%
	\printfield{eid}}
\DeclareFieldFormat[article]{number}{\mkbibparens{#1}}
\DeclareFieldFormat[article]{volume}{\textbf{#1}}
\DeclareFieldFormat{year}{\mkbibparens{#1}}
\DeclareBibliographyDriver{article}{%
	\printnames{author}:%
	\newunit\newblock
	\printfield{title}%
	\newunit\newblock
	\printfield{journaltitle}
	\newunit
	\iffieldundef{number}{\printfield{volume}}{\printfield{volume}\addspace\printfield{number}}
	\addcomma\addspace\printfield{pages}\addspace
	\printfield{year}
}
\AtEveryBibitem{\clearfield{month}}
\AtEveryBibitem{\clearfield{day}}
\usepackage[english]{babel}
\usepackage[T1]{fontenc}
\usepackage{csquotes}
\usepackage{bbm}
\usepackage{amsmath,amsfonts,amsthm,amsbsy,amssymb,dsfont,stmaryrd}
\usepackage[dvipsnames]{xcolor}
\usepackage{braket}

\usepackage{mathtools}

\numberwithin{equation}{section}

\newcommand\myshade{85}
\colorlet{mylinkcolor}{violet}
\colorlet{mycitecolor}{YellowOrange}
\colorlet{myurlcolor}{Aquamarine}

\usepackage[left=2.5cm,right=2.5cm,top=2.5cm,bottom=2.5cm]{geometry}
\usepackage[unicode=true,pdfusetitle,
bookmarks=true,bookmarksnumbered=false,bookmarksopen=false,
breaklinks=false,pdfborder={0 0 1},backref=false,
linkcolor  = mylinkcolor!\myshade!black,
citecolor  = mycitecolor!\myshade!black,
urlcolor   = myurlcolor!\myshade!black,
colorlinks = true,
]
{hyperref}
\usepackage[nameinlink]{cleveref}

\usepackage{tikz}
\usetikzlibrary{arrows}
\usetikzlibrary{intersections}

\usepackage{graphicx}
\usepackage{caption}

\definecolor{ct_black}{HTML}{000000}
\definecolor{ct_orange}{HTML}{ED872D}
\definecolor{ct_purple}{HTML}{7A68A6}
\definecolor{ct_blue}{HTML}{348ABD}
\definecolor{ct_turquoise}{HTML}{188487}
\definecolor{ct_red}{HTML}{E32636}
\definecolor{ct_pink}{HTML}{CF4457}
\definecolor{ct_green}{HTML}{467821}

\definecolor{ct2_green}{HTML}{9FF781}
\definecolor{ct2_green_dark}{HTML}{088A08}

\theoremstyle{plain}
\newtheorem{thm}{\protect\theoremname}[section]
\theoremstyle{plain}
\newtheorem{lem}[thm]{\protect\lemmaname}
\theoremstyle{plain}
\newtheorem{cor}[thm]{\protect\corollaryname}
\theoremstyle{plain}
\newtheorem{prop}[thm]{\protect\propositionname}
\theoremstyle{remark}

\theoremstyle{remark}
\newtheorem{rem}[thm]{\protect\remarkname}
\theoremstyle{definition}
\newtheorem{defn}[thm]{\protect\definitionname}
\theoremstyle{plain}

  \providecommand{\assumptionname}{Assumption}
\providecommand{\claimname}{Claim}
\providecommand{\corollaryname}{Corollary}
\providecommand{\definitionname}{Definition}
\providecommand{\lemmaname}{Lemma}
\providecommand{\propositionname}{Proposition}
\providecommand{\remarkname}{Remark}
\providecommand{\theoremname}{Theorem}
\providecommand{\examplename}{Example}

\crefname{section}{Section}{Sections}
\crefname{appendix}{Appendix}{Appendices}
\crefname{figure}{Figure}{Figures}
\crefname{assumption}{Assumption}{Assumptions}
\crefname{thm}{Theorem}{Theorems}
\crefname{lem}{Lemma}{Lemmas}
\crefformat{equation}{(#2#1#3)}

\crefrangelabelformat{equation}{(#3#1#4--#5#2#6)}

\crefmultiformat{equation}{(#2#1#3}{, #2#1#3)}{#2#1#3}{#2#1#3}
\Crefmultiformat{equation}{(#2#1#3}{, #2#1#3)}{#2#1#3}{#2#1#3}

\newtheorem*{lem*}{\protect\lemmaname}

\newcommand{\LOC}[1]{\mathrm{LOC}_{#1}}
\newcommand{\ee}{\operatorname{e}}
\newcommand{\ii}{\operatorname{i}}
\newcommand{\Mat}{\operatorname{Mat}}
\newcommand{\ZZ}{\mathbb{Z}}
\newcommand{\TT}{\mathbb{T}}

\newcommand{\NN}{\mathbb{N}}
\newcommand{\RR}{\mathbb{R}}
\newcommand{\CC}{\mathbb{C}}

\newcommand{\calF}{\mathcal{F}}

\newcommand{\Compacts}[1]{\mathcal{K}(#1)}
\newcommand{\Fredholms}[1]{\mathcal{F}(#1)}

\newcommand{\norm}[1]{\left\|#1\right\|}

\newcommand{\dif}{\operatorname{d}}
\newcommand{\tr}{\operatorname{tr}}
\renewcommand{\Im}{\operatorname{\mathbb{I}\mathbbm{m}}}
\renewcommand{\Re}{\operatorname{\mathbb{R}\mathbbm{e}}}
\newcommand{\ve}{\varepsilon}
\newcommand{\Id}{\mathds{1}}
\newcommand{\HH}{\mathcal{H}}

\newcommand{\BLO}[1]{\mathcal{B}(#1)}

\newcommand{\Closed}[1]{\mathrm{Closed}(#1)}
\newcommand{\dist}[1]{\mathrm{dist}(#1)}

\newcommand{\findex}{\operatorname{ind}}

\newcommand{\coker}{\operatorname{coker}}
\newcommand{\supp}{\operatorname{supp}}

\newcommand{\im}{\operatorname{im}}

\title{Two-Dimensional Time-Reversal-Invariant\\Topological Insulators via Fredholm Theory}
\author{Eli Fonseca, Jacob Shapiro, Ahmed Sheta, Angela Wang, Kohtaro Yamakawa\\\footnotesize{Mathematics Department, Columbia University, New York, NY 10027,
		USA}}

\begin{document}
	
\maketitle

\begin{abstract}
We study spinful non-interacting electrons moving in two-dimensional materials which exhibit a spectral gap about the Fermi energy as well as time-reversal invariance. Using Fredholm theory we revisit the (known) bulk topological invariant, define a new one for the edge, and show their equivalence (the bulk-edge correspondence) via homotopy.
\end{abstract}
 
\section{Introduction}

Insulators in two space dimensions obeying fermionic time-reversal symmetry \cite[Class AI]{AltlandZirnbauer97} have two distinct topological phases \cite{Schnyder_Ryu_Furusaki_Ludwig_PhysRevB.78.195125,Schnyder_Ryu_Furusaki_Ludwig_1367-2630-12-6-065010,Kitaev2009}. The fact there exists a non-trivial phase for such systems had not been immediately clear since the integer quantum Hall effect (IQHE) \cite{vonKlitzing_1979_PhysRevLett.45.494} is always trivial in the presence of time-reversal symmetry. In pioneering studies on Hall fluids in the early 1990s Fr\"ohlich and Studer \cite{Frohlich_Studer_1993_RevModPhys.65.733} discovered that despite time-reversal symmetry, such systems may exhibit non-trivial effects; they used Chern-Simons effective field theory. Much later, this non-triviality was rediscovered in \cite{Kane_Mele_2005,Fu_Kane_2007,Bernevig1757}, now from the perspective of single-particle translation invariant Hamiltonians, and experimental investigations \cite{doi:10.1143/JPSJ.77.031007,Roth294,Brune2010,pub.1035960235,Hsieh2008} followed. Physically, the topological non-triviality of such systems is associated with an unpaired state at the boundary of the sample.

Via the discovery of \cite{TKNN_1982} there came an association with the field of algebraic topology. Mathematically, in the presence of translation invariance for bulk systems, bulk insulators have associated with them a $\CC$-vector-bundle over $\TT^d$, the Brillouin zone. The topological classes of such vector bundles are studied via the well-known Chern characteristic classes \cite{Milnor_Willard_Stasheff_1974}. In the case $d=2$ one obtains an isomorphism of all classes with $\ZZ$ via the Chern number. In the presence of time-reversal symmetry, as mentioned already, the Chern number is always zero. However, time-reversal which squares to $-\Id$ defines a \emph{quaternionic structure} on such vector bundles where now the Pontryagin classes may be non-trivial \cite{Nittis_Gomi_2016}. The breakthrough study of \cite{Bellissard_1994} allowed one to do away with the assumption of translation invariance, which was crucial to explain the IQHE; one uses K-theory of C-star algebras as the main algebraic tool, and \emph{index theorems} relating physical quantities to indices of Fredholm operators guarantee topological properties.

A major theme in the study of topological insulators is \emph{the bulk-edge correspondence}: the fact that topological invariants computed for an infinite bulk system agree with those computed from that system truncated to the half-infinite space--the edge system. Such proofs first emerged in the context of the IQHE \cite{Hatsugai_1993} in ascending degrees of generality \cite{kellendonk_richter_schulz-baldes_2002,Elbau_Graf_2002,EGS_2005,Taarabt_1403.7767,1901.06281} and then also for other symmetry classes, dimensions as well as time-periodic driven systems \cite{Graf2013,Graf_Shapiro_2018_1D_Chiral_BEC,Kubota2017,Bourne2017,Sadel2017,Graf2018,Shapiro2019}.

Here we study the various indices as well as the bulk-edge correspondence  by singling out the \emph{Fredholm theory of local operators} (Fredholm operators with off-diagonal decaying matrix elements in the position basis--see \cref{def:local operator}) as the most natural level of generality in the disordered spectral gap regime, where by natural we mean a certain optimum between length and generality of proofs. The idea is simple and applies to all cases of topological insulators. In \cite{Graf_Shapiro_2018_1D_Chiral_BEC} it was applied to chiral 1D systems which are possibly the simplest topological insulators, so it is worthwhile to demonstrate how it works in the much subtler 2D IQHE and time-reversal invariant cases. Informally stated, the idea is that the bulk topology is always computed via a flat Hamiltonian (the polar part in the polar decomposition, after properly shifting so that the gap contains zero) $$ H \rightsquigarrow H|H|^{-1} $$ which is equivalent to the Fermi projection in the spectral gap regime. The edge topology, conversely, is computed from the truncated (to the half-space) bulk Hamiltonian $$ H \rightsquigarrow \iota^\ast H \iota $$ where $\iota$ is the injection from the half-space to the full space. Our results show that as long as Hamiltonians are spectrally gapped and local, these two operations of flattening and truncating commute, as far as Fredholm theory can gauge. This is captured in the homotopy of equation \cref{eq:ultimate homotopy showing the commutation of truncation and flattening} below. This idea produces a new proof to the already existing abundance of bulk-edge proofs in the spectrally gapped IQHE case \cite{kellendonk_richter_schulz-baldes_2002,Elbau_Graf_2002}, though it is possibly shorter. However, we get a new proof when we take up the fermionic time-reversal invariant case, which is harder, previously studied by Graf and Porta in \cite{Graf2013} assuming translation invariance in one axis and nearest-neighbor hopping in the other. Here we generalize to general disordered spectrally gapped systems with general boundary conditions. Our bulk $\ZZ_2$ invariant was already defined in this disordered spectrally gapped context by Schulz-Baldes \cite{Schulz-Baldes_2015_Z2} and then again by Katsura and Koma \cite{Katsura_Koma_2016} (the two are equivalent) which are equivalent to the Fu-Kane-Mele invariant \cite{Fu_Kane_2007} and the Graf-Porta invariant \cite{Graf2013} when the latter two are defined. Our edge invariant follows ideas from \cite{kellendonk_richter_schulz-baldes_2002}. 

We remark that in the \emph{strongly} disordered mobility gap case \cite{EGS_2005}, the question of how to define the edge invariant for time-reversal invariant systems remains open, as does the bulk-edge correspondence proof. The main hurdle seems to be the absence of a trace formula for the $\ZZ_2$ invariant (an analog of the Fedosov formula for Fredholm operators; cf. \cref{thm:local Z_2 trace formula}), or alternatively the absence of a regularization of the edge operator (see \cref{eq:edge Fredholm operator} below) which would make it Fredholm also in the mobility gap regime. Should the latter be achieved, it is not inconceivable that a homotopy proof could be extended to handle strong disorder.

After briefly discussing spectrally gapped \emph{local} operators, we define the bulk and edge topological invariants and state our main theorem regarding their equivalence. We then digress shortly to discuss some consequences of such non-trivial systems. Finally we give the proofs of the correspondence theorem in the following section. The appendix contains some discussion of the $\ZZ_2$-valued Fredholm index for convenience of the reader unfamiliar with \cite{Atiyah1969,Schulz-Baldes_2015_Z2}.

\section{Setting and main result}

Our single particle bulk Hilbert space is $\HH:=\ell^2(\ZZ^2)\otimes\CC^N$ for some fixed (once and for all) $N\in\NN_{\geq1}$, and the half-infinite edge Hilbert space is $\hat{\HH} := \ell^2(\ZZ\times\NN)\otimes\CC^N$. Throughout, edge objects will carry a hat. We have the natural injection $$ \iota:\ell^2(\ZZ\times\NN)\hookrightarrow\ell^2(\ZZ^2)$$ which extends to the descendant spaces as well. It extends a wave function on the half-space by zero, and its adjoint $\iota^\ast$ is truncation to the half-space. We have \begin{align}|\iota|^2 \equiv\iota^\ast\iota= \Id_{\hat{\HH}}\,,\qquad|\iota^\ast|^2 \equiv\iota\iota^\ast= \Lambda(X_2)\,,\qquad(\Id-\Lambda(X_2))\iota=0\label{eq:injection and truncation relations}\end{align} with $\Lambda$ the Heaviside step function (or a finite perturbation of it) and $X_j$ the position operator in direction $j=1,2$.

\paragraph{Spatial constraints.}
On $\HH$, we consider operators $A\in\BLO{\HH}$ where often we will refer to their matrix elements $$A_{xy} := \langle\delta_x, A\delta_y\rangle\in\Mat_N(\CC)$$ where $\Set{\delta_x}_{x\in\ZZ^2}$ is the canonical position basis. By $\|A_{xy}\|$ we mean the trace norm on $\Mat_N(\CC)$. 

We will also often need spatial constraints on such matrix elements which encode the principle that the laws of physics act locally in space, and we adopt the terminology of \cite[Section 3.1]{Shapiro2019}:
\begin{defn}[Local operator] The operator $A$ is called (any-rate-polynomially) local iff for any $\alpha\in\NN$ sufficiently large there is some $C_\alpha<\infty$ such that \begin{align*}
	\|A_{xy}\| \leq C_\alpha (1+\|x-y\|)^{-\alpha}\qquad(x,y\in\ZZ^2)\,.
\end{align*} 
\label{def:local operator}

If the term local appears alone, then the polynomial rate is meant. However, sometimes we also require \emph{exponential locality}, which we define as the existence of constants $C<\infty,\mu>0$ such that \begin{align*}
\|A_{xy}\| \leq C \exp(-\mu\|x-y\|)\qquad(x,y\in\ZZ^2)\,.
\end{align*} 
\end{defn}

We recall that in \cite[Section 3.1]{Shapiro2019} it was established that the \emph{weakly local} operators (those operators where the rate of off-diagonal decay is not diagonally uniform) form a star-algebra, and the same proofs go through for the space of local operators (be it exponential or polynomial), a fact we shall freely use. Furthermore, the smooth functional calculus of (exponentially) local operators is (polynomially) local \cite[Appendix A]{Elbau_Graf_2002}, a fact crucial to our analysis. See also \cref{section:control_of_functional_calculus_of_difference_of_operators}. 

When dealing with half-space objects, we shall need the definition 
\begin{defn}[Local and confined operator]\label{def:local and confined operator} The operator $A\in\BLO{\HH}$ is local and confined in direction $j=1,2$ iff for any $\alpha\in\NN$ sufficiently large there is some $C_\alpha<\infty$ such that \begin{align*}
	\|A_{xy}\| \leq C_\alpha(1+\|x-y\|)^{-\alpha}(1+|x_j|)^{-\alpha}\qquad(x,y\in\ZZ^2)\,.
	\end{align*}
\end{defn}

This notion makes sense also for the edge Hilbert space, with obvious modifications.

We also recall that in \cite[Section 3.3]{Shapiro2019} it was established that the local and confined operators form a star-closed two-sided ideal within the star-algebra of local operators, that if $A$ is local and confined in direction $1$ and $B$ is local and confined in direction $2$ then $AB$ is local and confined simultaneously in direction $1$ and $2$ and that such operators which are confined in all space directions are trace class.

We denote by $\partial_j$ the non-commutative derivative in direction $j$, which is defined as $$ \partial_j A := -\ii[\Lambda(X_j),A] $$ $\Lambda$ as in \cref{eq:injection and truncation relations}. According to \cite[Corollary 3.16]{Shapiro2019}, if $A$ is local then $\partial_j A$ is local and confined in direction $j$.

\begin{defn}[$\LOC{2}$]
	Since the $2$-direction has a special significance in this paper, often we will abbreviate "local and confined in the $2$-direction" as \emph{$\LOC{2}$}.
\end{defn}

\paragraph{Insulators.}

Though not the most general, in the present study we will content ourselves with the following 
\begin{defn}[Insulator] An insulator is a self-adjoint, \emph{exponentially} local Hamiltonian $H\in\BLO{H}$ such that $$0\notin\sigma(H)\,.$$
	
Since $\sigma(H)\in\Closed{\CC}$, there is an entire open interval about zero in $\RR\setminus\sigma(H)$ which we call \emph{the spectral gap}.
\end{defn}
We note that our methods do not (yet) deal with mobility-gapped insulators (where zero \emph{is} within the spectrum, but the states about it are dynamically localized \cite{EGS_2005}) though the bulk (but not the edge) invariant we use remains well-defined also in that regime. For this reason most of what we need to deal with are local rather than weakly-local operators.

Time-reversal symmetry is an operator $\Theta$ defined as anti-unitary $\Theta:\HH\to\HH$ (so it is anti-$\CC$-linear and $\Theta^{\ast}\Theta=\Id$) such that $[\Theta,X]=0$ and $\Theta^2=-\Id$. This latter constraint is due to the fact we are describing electrons, which are Fermions.

\begin{defn}[Time-reversal invariant insulators]
	An insulator $H\in\BLO{\HH}$ is time-reversal invariant iff $$ [H,\Theta]=0\,.$$
\end{defn}

\paragraph{The bulk topological invariant.} As in \cite{ASS1994_charge_def}, we shall use the unitary operator implementing the gauge transformation associated with a magnetic flux insertion at the origin given by $$ U := \exp(\ii\arg(X_1+\ii X_2)) $$ where $X_j$ is the position operator on $\ell^2(\ZZ^2)$ in the $j=1,2$ direction.

Then, with $ P := \chi_{(-\infty,0)}(H)$ being the Fermi projection corresponding to the insulator $H$ (which implies $P$ is local, as due to the spectral gap assumption it is a smooth function of $H$), we know \cite{Bellissard_1994} that the following operator is Fredholm \begin{align*} F:= \mathbb{P}U \end{align*} where we define the super operator \begin{align}\mathbb{Q}A:= Q A Q + Q^\perp\cong \left.QAQ\right|_{\im(Q)}\label{eq:projection truncation super operator}\end{align} for any projection $Q$ and any operator $A$ (here $Q^\perp\equiv\Id-Q$) and the Hall conductivity in two-dimensions equals $$\findex F\equiv\dim\ker F - \dim\ker F^\ast\,.$$ We remark in passing that $\mathbb{Q}$ is idempotent as well.

When $H$ is time-reversal invariant, $[P,\Theta]=0$ as well, but due to anti-linearity and $[\Theta,X]=0$, $\Theta U \Theta = - U^\ast$ so that \begin{align}
	F = -\Theta F^\ast \Theta
\label{eq:the theta-odd constraint}\end{align}  
which implies that $\findex F = 0$ with the logarithmic law of the Fredholm index \cite{Booss_Topology_and_Analysis} and the fact that $\pm\Theta$ is invertible.

Despite the Hall conductivity always being zero when time-reversal invariance is present, another index exists for Fredholm operators obeying \cref{eq:the theta-odd constraint}, which was already studied in \cite{Atiyah1969} under the name "skew-adjoint Fredholm operators". In \cite{Schulz-Baldes_2015_Z2} this index was first applied to describe time-reversal invariant bulk systems.

\begin{defn}[$\Theta$-odd Fredholm operators] $A\in\BLO{\HH}$ is a $\Theta$-odd Fredholm operator iff $A$ is Fredholm and obeys \cref{eq:the theta-odd constraint}. The space of all such operators is denoted by $\calF_\Theta(\HH)$, which is a subspace of the space of Fredholm operators $\calF(\HH)$ in which all operators have zero Fredholm index.
\end{defn}

\begin{rem}
	This notion is equivalent to Atiyah and Singer's skew-adjoint Fredholm operators \cite{Atiyah1969}, to Schulz-Baldes' odd-symmetric operators \cite{Schulz-Baldes_2015_Z2} as well as to Katsura and Koma's $\ZZ_2$-index of pair of projections \cite{Katsura_Koma_2016}. We prefer the name \emph{$\Theta$-odd} because it seems most natural to us to work with the already-provided real structure given from physics by $\Theta$.
\end{rem}

\begin{defn}
	We define the $\Theta$-odd Fredholm index of any $A\in\calF_\Theta(\HH)$, after Atiyah and Singer \cite{Atiyah1969}, as $$ \findex_2 A := \dim \ker A \mod 2\,. $$
\end{defn}

It is a fact that similar to the Fredholm index, $\findex_2$ is stable under norm continuous and compact perturbations which preserve the $\Theta$-odd condition \cref{eq:the theta-odd constraint}. Proven already both in \cite{Atiyah1969} and in \cite{Schulz-Baldes_2015_Z2}, we include a brief discussion of these facts in \cref{subsec:theta-odd Fredholm theory} using our $\Theta$-odd terminology for the reader's convenience.

As such, it makes sense to define for time-reversal invariant insulators $H$ the topological invariant as $$ \mathcal{N} :=  \findex_2 F\,. $$

The proof that the various formulations of the index are equivalent is the same: each invariant is $\ZZ_2$-valued and homotopy stable, and hence, if it agrees with another invariant on elements from the trivial as well as the non-trivial class, then the two invariants are identical. We demonstrate this in \cref{sec:equivalence of z_2 indices}.

\paragraph{The edge topological invariant.} As far as we know, our definition of the edge invariant in the current level of generality is actually new. Indeed, usually (in the physics literature \cite{Hasan_Kane_2010}) the edge invariant is defined only assuming translation invariance in the bulk, in which case it is counted as a $\ZZ_2$ spectral flow: the mod $2$ signed number of crossings of zero of the edge spectrum as the (still valid) momentum variable traverses half its range, $[0,\pi]$. See also the edge invariant of \cite{Graf2013}.

Here instead we adapt the IQHE edge invariant of \cite{kellendonk_richter_schulz-baldes_2002}: In the IQHE, Kellendonk, Richter and Schulz-Baldes show that the edge Hall conductivity (a tracial formula which is the expectation value of the velocity operator along the boundary for the edge states) is given by a Fredholm index as well (i.e., it is an index theorem proving both the quantization of the edge Hall conductivity as well the bulk-edge correspondence). 

Let us setup the edge picture. We already noted that the Hilbert space for the edge is $\hat{\HH}\equiv\ell^2(\ZZ\times\NN)\otimes\CC^N$. Given a bulk Hamiltonian $H\in\BLO{\HH}$, a natural edge Hamiltonian is induced as $\operatorname{Ad}_{\iota^\ast}H:=\iota^\ast H \iota$. This definition corresponds to \emph{Dirichlet} boundary conditions. More general conditions may be implemented by adding any self-adjoint (exponentially) $\LOC{2}$ operator, so that all in all we assume that our edge Hamiltonian $\hat{H}\in\BLO{\hat{H}}$ is self-adjoint, local, \emph{not} an insulator, but stems from some bulk insulator. The relationship with a bulk Hamiltonian $H$ is made through the constraint that they are compatible:
\begin{defn}\label{def:compatibility between bulk and edge}
	An edge Hamiltonian $\hat{H}\in\BLO{\hat{\HH}}$ is \emph{compatible} with a bulk Hamiltonian $H\in\BLO{\HH}$ iff $(\operatorname{Ad}_{\iota^\ast}H)-\hat{H}$ is exponentially $\LOC{2}$.
\end{defn}

\begin{rem}
	We note that since we require $H$ to be exponentially local (which we need to apply the Combes-Thomas estimate, \cref{thm:Combes-Thomas estimate}, later on), it is only natural to ask that the boundary conditions be exponentially local as well, rather than polynomial.
\end{rem}
Then \cite{kellendonk_richter_schulz-baldes_2002,PSB_2016} define the edge IQHE Fredholm index as follows. Let $g:\RR\to[0,1]$ be a smooth function interpolating between $1$ on the left and $0$ on the right such that $g'$ is supported within the bulk gap of $H$. Then we define \begin{align} \hat{F} := 
\mathbb{W}_1 g(\hat{H})
\label{eq:edge Fredholm operator}\end{align}
where we define the winding operator in direction-1 as \begin{align}\mathbb{W}_1 A := \mathbb{\Lambda}_1\exp(-2\pi\ii A) \label{eq:winding super operator}\end{align} for any operator $A$, where we use the notation $\Lambda_j := \Lambda(X_j)$ with $\Lambda$ as in \cref{eq:injection and truncation relations} (in conjunction with $\mathbb{\Lambda}_1 A\equiv \Lambda_1 A\Lambda_1+\Lambda_1^\perp$ as in \cref{eq:projection truncation super operator}). We remark that the factor $2\pi \ii$ is \emph{not} fixed by convention but is rather necessary, as we shall see later.

\cite[Proposition 7.1.2]{PSB_2016} (see also references therein) then proves $\findex \hat{F}$ equals the edge Hall conductivity. We note here the (heuristic) reason that $\hat{F}$ is Fredholm is because it is composed of two parts: its bulk part is essentially a projection (and $\exp(-2\pi\ii \mathrm{projection}) = \Id$), and its edge part is confined in the $2$ direction. This is proven also in \cite{Elbau_Graf_2002}, but see below for more details.

Our observation here is that even though in the time-reversal invariant case there is no meaning to the edge Hall conductivity any longer, $\hat{F}$ still makes sense, and in fact, obeys \cref{eq:the theta-odd constraint} too. Hence, $\hat{F}\in\calF_\Theta(\hat{\HH})$ so that we may take an a-priori definition of the edge $\ZZ_2$ index as $$ \hat{\mathcal{N}} := \findex_2 \hat{F}\,.$$

We note $\hat{\mathcal{N}}$ makes sense as long as there is a spectral gap, but \emph{not} in the mobility gap regime, unlike $\mathcal{N}$ which remains valid also in the latter regime.

\paragraph{Main result.}
\begin{thm}\label{thm:BEC}
	For any bulk time-reversal invariant insulator $H$ and a compatible time-reversal invariant edge Hamiltonian $\hat{H}$ (in the sense of \cref{def:compatibility between bulk and edge}) we have $$ \mathcal{N}=\hat{\mathcal{N}}\,. $$
	
\end{thm}

\begin{rem}
	Actually since our method of proof uses interpolation and homotopy alone, it applies also in the Fredholm case and not only in the $\Theta$-odd Fredholm case, so that it actually reproduces existing bulk-edge proofs in the spectral gap regime of the IQHE \cite{kellendonk_richter_schulz-baldes_2002,Elbau_Graf_2002} which use K-theory or functional analysis respectively. Our reproduction uses Fredholm theory alone.
\end{rem}
\subsection{Discussion of consequences}
\paragraph{Continuity properties.}
In the spectral gap regime, if two insulators $H$ and $H'$ have a small norm difference $\|H-H'\|$ (we always assume they are spectrally gapped at zero, so by construction they have a common gap) then using the Riesz projection and the resolvent identity we see that $\|P-P'\|$ is also small. Indeed, if $\Gamma$ is a path in $\CC$ encircling both $\sigma(H)$ and $\sigma(H')$ without crossing these sets, \begin{align}P-P'&=\frac{\ii}{2\pi}\oint_\Gamma (R(z)-R'(z))\dif{z}\\&=\frac{\ii}{2\pi}\oint_\Gamma R(z)(H'-H)R'(z)\dif{z}\,.\end{align} Then one can use the trivial bound $\|R(z)\|\leq\dist{z,\sigma(H)}^{-1}$ and the fact $|\Gamma|<\infty$. This clearly implies $\|F-F'\|$ is small, so by stability of $\findex_2$ with respect to norm continuous perturbations \cref{thm:continuity of index2}, we see that $\mathcal{N}$ is locally constant. \cref{thm:BEC} then implies the stability of $\hat{\mathcal{N}}$ too.

\paragraph{Local $\Theta$-Odd Index Formula.}
One major theme in the IQHE is the fact that the Hall conductivity is a \emph{linear response coefficient}, i.e., it is the system's (linear) response to external driving (electric field). While it is not clear to us whether the $\ZZ_2$ topological invariant is such a response of the system as well (we conjecture it is not), one thing that usually comes out of phrasing the quantity as a response coefficient is a tracial, local formula (for example, for the IQHE one has the Kubo formula \cref{eq:bulk index}). Such a local trace formula has recently been suggested in \cite{Li_Mong_1905.12649}, where it is claimed to work in the mobility gap regime as well; see also \cite{Lozano_et_al_2019} for another perspective.

Here we propose an alternative local trace formula for $\findex_2$, which however entails a limiting process.

Starting from the Fedosov formula \cite{Calderon_1967}, one has that for any Fredholm operator $A$, \begin{align} \findex A = \tr((\Id-|A|^2)^n)-\tr((\Id-|A^\ast|^2)^n) \label{eq:Fedosov formula}\end{align}
if an $n\in\NN_\geq1$ exists such that the operators within the traces are actually trace class (which is always the case in our applications). Since we know the spectrums of $|A|^2$ and $|A^\ast|^2$ agree except for the kernels, and since the kernels coincide with the kernels of $A,A^\ast$ respectively, the traces yield $\pm1^n$ on vectors within the kernels of $A,A^\ast$ respectively, and otherwise all other contributions are equal between the two traces and thus cancel out. If $\Id-|A|^2\leq1$ (which is certainly the case of interest here for $F$ and $\hat{F}$, which both have norm smaller than or equal to $1$) then all these additional contributions which should cancel out are within $(0,1)$, before raising them to the power $n$. Indeed, $\Id-|A|^2$ is Schatten and hence compact, as well as self-adjoint, and therefore has an ONB. Taking the trace in that basis whose eigenvalues are $\{\lambda_j\}_j$, we find \begin{align} \tr((\Id-|A|^2)^n) &= \sum_{j\in\NN:\lambda_j=1}\lambda_j^n+\sum_{j\in\NN:\lambda_j<1} \lambda_j^n \\ &=\dim \ker A+\sum_{j\in\NN:\lambda_j<1} \lambda_j^n \end{align}

We note that since $(\Id-|A|^2)^n$ is trace-class, this last sum converges indeed. Furthermore, using Lebesgue's monotone convergence theorem we can evaluate $n\to\infty$ of both sides of this equation. We have thus proven

\begin{thm}\label{thm:local Z_2 trace formula}
	Let $A\in\mathcal{F}_\Theta(\HH)$ with $\|A\|\leq1$ such that $\Id-|A|^2,\Id-|A^\ast|^2$ are of Schatten class. We have $$ \lim_{n\to\infty} \tr((\Id-|A|^2)^n) = \dim \ker A $$ and so a local formula for $\findex_2(A)$ when $\|A\|\leq1$ is $$\findex_2 A = \left(\lim_{n\to\infty} \tr((\Id-|A|^2)^n)\right)\mod2\,.$$
\end{thm}

\paragraph{Complete localization in two dimensions.}
The one parameter scaling theory of localization of Abrahams, Anderson, Licciardello and Ramakrishnan \cite{Anderson_1979_PhysRevLett.42.673} predicts that disordered two dimensional materials \emph{obeying time-reversal symmetry} exhibit no metal-insulator transition (so this excludes IQHE systems, with strong magnetic field and breaking of time-reversal symmetry). Here we discuss how the existence of non-trivial $\ZZ_2$ topological systems stands in mathematical contradiction with this complete 2D localization, see also the experimental study in \cite{Su_et_al_2016}.

While there are various ways to approach mathematical Anderson localization \cite{Frohlich1983}, one very important criterion, demonstrated in \cite{Aizenman_Graf_1998} is that if a system is Anderson localized at energy $\mu\in\RR$ then its Fermi projection at $\mu$, i.e., $P_\mu \equiv \chi_{(-\infty,\mu)}(H)$, has off-diagonal exponential decay (albeit not-uniformly): There exists some $\nu>0$ such that for any $\ve>0$ there exists some $C_\ve<\infty$ such that \begin{align} \|(P_\mu)_{xy}\| \leq C_{\ve}\exp(-\nu\|x-y\|+\ve\|x\|)\qquad(x,y\in\ZZ^2)\,.
\label{eq:weakly-local condition for the Fermi projection}
\end{align}
This decay in fact is what guarantees that our operator $F_\mu\equiv P_\mu U P_\mu + P_\mu^\perp$ is Fredholm. 

Now let $H$ be some two-dimensional time-reversal symmetric system such that $F\equiv F_0\in\mathcal{F}_{\Theta}(\HH)$ and $\findex_2(F)=1$. Since we know that placing the Fermi energy below the spectrum of $H$, we have $F_{-2\|H\|} = \Id$ and so $\findex_2(F_{-2\|H\|})=0$. It follows that there must be a point $\mu_c\in(-2\|H\|,0)$ where $F_{\mu_c}$ was \emph{not} Fredholm, and so \cref{eq:weakly-local condition for the Fermi projection} must fail for $\mu=\mu_c$. Since we know \cref{eq:weakly-local condition for the Fermi projection} is implied by Anderson localization as defined in, e.g. \cite{Aizenman_Graf_1998}, it is impossible that the system was Anderson localized at $\mu_c$.

This phenomenon was studied for the random Landau Hamiltonian quantitatively in \cite{Germinet_Klein_Schenker_2004}, where it was proven there must be a point of "delocalization" somewhere within each Landau smeared (due to disorder) band. That point is precisely necessary due to the non-zero Chern number of each Landau level.

\section{Correspondence proofs}
The first step in our proof is to connect $\mathcal{N}$ to Kitaev's bulk Fredholm index \cite[eq-n (131)]{kitaev2006}. The point is that $F$ is associated with a flux insertion at the origin, which corresponds to radial geometry as in Laughlin's IQHE explanation (see \cite{ASS1994_charge_def} and references therein). However, $\hat{F}$ actually corresponds to a geometry of a corner at the origin within the first quadrant of the plane $\ZZ^2$. Fortunately, Kitaev \cite{kitaev2006} worked out the bulk Fredholm index also for the rectangular geometry out of the Kubo formula. Using the fact that in the spectral gap regime $P=g(H)$ (already anticipating the connection with the edge index), we have 
\begin{thm}\label{thm:equivalence of bulk indices}
	Even without time-reversal invariance, we have ($\mathbb{W}_1$ as in \cref{eq:winding super operator}) \begin{align}\findex F = \findex\mathbb{W}_1( g(H)\Lambda_2 g(H))\,.\label{eq:relationship between F and Kitaev bulk}\end{align}
	The connection between the two is made via the Fedosov formula and then the Kubo formula  \begin{align}\label{eq:bulk index} 2\pi\sigma_\mathrm{Hall}(H) = -2\pi\ii\tr P[\partial_1 P,\partial_2 P]\,.\end{align}  
	Furthermore, we also have independently of the above \begin{align} \mathcal{N} = \findex_2\mathbb{W}_1( g(H)\Lambda_2 g(H))\label{eq:equivalence of bulk z2 indices}\end{align} in the time-reversal invariant case.
\end{thm}
\begin{proof} We postpone the proof of the first statement to the appendix below, since its use is only made in \cref{prop:Z2 index equivalence}.
	
To see that our definition of $\mathcal{N}$ is equivalent \emph{any} other $\ZZ_2$ homotopy invariant (and hence prove \cref{eq:equivalence of bulk z2 indices}), it suffices to show that our invariant agrees with any other on elements from the trivial as well as the non-trivial class because each invariant is $\ZZ_2$-valued and homotopy stable. This is done in \cref{prop:Z2 index equivalence}.
	 \end{proof}

Next, it would be nice if we could replace $g(H)\Lambda_2 g(H)$ with $\Lambda_2 g(H) \Lambda_2$ in \cref{eq:equivalence of bulk z2 indices}, to make it look more like $\hat{F}$ in \cref{eq:edge Fredholm operator}. This should be a an almost true statement away from the boundary, where $\partial_2 g(H)\approx0$ and as $g(H),\Lambda_2$ are projections. In order to achieve that we setup a homotopy. 

To prove this and other homotopies pass within $\calF(\HH)$, we use the following heuristic idea: $\exp(-2\pi\ii Q)-\Id=0$ for any projection $Q$. If $Q$ is not truly a projection, but its difference from a projection is confined in the $2$-direction, then $\exp(-2\pi\ii Q)-\Id$ is also confined in the $2$-direction. This is the contents of \cref{prop: A^2-A LOC2 then the derivative is compact}.

\begin{lem}
	There is a homotopy within $\calF(\HH)$ from $$\mathbb{W}_1( g(H)\Lambda_2 g(H)) \longrightarrow \mathbb{W}_1( \Lambda_2 g(H) \Lambda_2)$$ and if $[H,\Theta]=0$ then it passes within $\calF_\Theta(\HH)$, so that both $\findex$ and $\findex_2$ agree on these two operators.
\end{lem}
\begin{proof}
		Consider the continuous map $V:[0,1] \to \BLO{\HH}$ defined by $[0,1]\ni t\mapsto \exp(-2\pi \ii A(t))$ with $A(t) = t\Lambda_2g(H)\Lambda_2+(1-t)g(H)\Lambda_2g(H)$.  Noting that both $g(H)$ and $\Lambda_2$ are projections after some algebra we have 
		\begin{align*}
			A(t)^2-A(t) &= t^2[\Lambda_2 g(H), g(H)\Lambda_2]^2+t[\Lambda_2 g(H), g(H)\Lambda_2]g(H)\Lambda_2 g(H)+\\&+tg(H)\Lambda_2g(H)[\Lambda_2 g(H), g(H)\Lambda_2]+g(H)[\Lambda_2 g(H), g(H)\Lambda_2]g(H)\,.
		\end{align*}
		Now $g(H)$ and $H$ are local operators and this latter expression is a polynomial in $[\Lambda_2 g(H), g(H)\Lambda_2]$, which may expressed as a sum of commutators of $\Lambda_2$ with a local operators, with local operators as coefficients. Using the algebraic closure properties of $\LOC{2}$ operators we conclude that $A(t)^2-A(t)$ is $\LOC{2}$ for all $t\in [0,1]$.  Hence by application of \cref{eq:key lemma to show homotopies are Fredholm} we conclude that the necessary homotopy is constructed. Note that $A(t)$ is self-adjoint so if $H$ is TRI then $[A(t), \Theta]=0$ so this homotopy will also pass within $\calF_\Theta(\HH)$.   
\end{proof}

Now, we manipulate the edge Fredholm operator $\hat{F}$: currently it is calculated via the dimension of a kernel within $\hat{\HH}$. To make a homotopy between it and bulk objects, we want to calculate it as the dimension of a kernel in $\HH$ instead. So we use the following
\begin{lem}
	We have \begin{align}\ker_{\hat{\HH}} \mathbb{W}_1(\pm g(\hat{H})) \cong \ker_{\HH} \mathbb{W}_1(\pm \Lambda_2 g(\Lambda_2 H \Lambda_2) \Lambda_2) \label{eq:connection between edge Fredholm index and bulk Fredholm index}\end{align}
	and hence both $\findex$ and $\findex_2$ for the edge may be calculated using the operators within the kernel in the RHS of \cref{eq:connection between edge Fredholm index and bulk Fredholm index}.
\end{lem}
\begin{proof}
Let us for the moment assume \emph{Dirichlet} boundary conditions, so $\hat{H} \equiv \operatorname{Ad}_{\iota^\ast} H$. In \cref{prop:general BC} we show how to generalize this. 

Given the edge Fredholm operator, $\hat{F}\equiv \mathbb{W}_1 g(\hat H)$, we begin by showing that $ \hat{\mathbb{F}}:= \iota \hat{F} \iota^* + \Lambda_2^{\perp}\in\BLO{\HH}$ is (1) also Fredholm and (2) has a kernel isomorphic to that of $\hat{F}$. (1) can be shown by constructing a parametrix of $\hat{\mathbb{F}}$. It is $\hat{\mathbb{G}}:=\iota \hat{G} \iota^\ast + \Lambda_2^\perp$, where $\hat{G}$ is the parametrix of $\hat{F}$. This follows after using the relations in \cref{eq:injection and truncation relations} and noting that conjugation by $\iota$ preserves compactness. 
(2) can be shown by noting that $\eta:\ker F \to \ker \hat{\mathbb{F}}$ defined as
$\eta:\hat \psi\mapsto \iota \hat \psi$ and  $\eta^{-1}: \psi\mapsto\iota^\ast \psi$ is a linear isomorphism (to see this we use again \cref{eq:injection and truncation relations} as well as the fact that $\iota$ is an injection, i.e. it has a trivial kernel). A similar argument may be made for $F^\ast$.
Thus,
\begin{align*}
\ker_{\hat{\HH}} \mathbb{W}_1(\pm g(\hat{H})) \cong \ker_{\HH}(\iota(\mathbb{W}_1( \pm g(\hat{H})))\iota^\ast + \Lambda_2^\perp)\,.
\end{align*}
Utilizing the fact that $\iota$ and $\iota^\ast$ commutes with $\Lambda_1$, the RHS simplifies to
	\begin{equation}
		\ker_{\HH}(\Lambda_1 \iota \exp(\pm2\pi\ii \iota^* g(\Lambda_2 H \Lambda_2) \iota) \iota^* \Lambda_1 +\Lambda_1 ^{\perp}\Lambda_2 + \Lambda_2^{\perp})
		\label{eq:intermediate edge kernel}
	\end{equation} 
where we used \cref{eq:injection and truncation relations} and $g(\hat{H}) = \iota^* g(\Lambda_2 H \Lambda_2) \iota$, which can be deduced from \cref{eq:the HS formula} and the fact that $$ (\hat{H}-z\Id_{\hat{\HH}})^{-1} = \iota^{\ast}(\Lambda_2 H \Lambda_2 - z \Id_{\HH})^{-1}\iota\,. $$  Now, since $(\iota^\ast g(\Lambda_2 H\Lambda_2)\iota)^n = \iota^\ast (\Lambda_2g(\Lambda_2H\Lambda_2)\Lambda_2)^n\iota$, if we expand the exponential as a series, we obtain $$\iota \exp(\pm2\pi\ii \iota^* g(\Lambda_2 H \Lambda_2)\iota) \iota^* =  \Lambda_2 \exp(\pm2\pi\ii \Lambda_2 g(\Lambda_2 H \Lambda_2)) \Lambda_2 = \exp(\pm2\pi\ii \Lambda_2 g(\Lambda_2 H \Lambda_2)\Lambda_2) - \Lambda_2^{\perp}$$ where the last equality follows because $\exp(\alpha QAQ) = \mathbb{Q}\exp(\alpha QAQ)$ (using \cref{eq:projection truncation super operator}), valid for any projection $Q$ and constant $\alpha$. The desired result is obtained by plugging this into \cref{eq:intermediate edge kernel}, which simplifies to 
$$\ker_{\HH}\mathbb{W}_1(\pm \Lambda_2 g(\Lambda_2 H \Lambda_2) \Lambda_2)$$
where we used $-\Lambda_1 \Lambda_2^{\perp} + \Lambda_1^{\perp}\Lambda_2 + \Lambda_2^{\perp} = \Lambda_1 ^{\perp}$.
\end{proof}

At this stage, the proof of \cref{thm:BEC} is completed with the construction of the following homotopy
\begin{lem}
	There is a homotopy within $\calF(\HH)$ from \begin{align} \mathbb{W}_1( \Lambda_2 g(H) \Lambda_2)  \longrightarrow \mathbb{W}_1( \Lambda_2 g(\Lambda_2 H \Lambda_2) \Lambda_2) \label{eq:ultimate homotopy showing the commutation of truncation and flattening}\end{align} and if $[H,\Theta]=0$ then it passes within $\calF_\Theta(\HH)$, so that both $\findex$ and $\findex_2$ agree on these two operators.
	\label{lem:final BEC homotopy}
\end{lem}
\begin{proof}

	Consider the continuous map $W:[0,1] \to \BLO{\HH}$ defined by $[0,1]\ni t \mapsto \exp(-2\pi \ii A(t))$
	 with $$A(t) = (1-t)\Lambda_2 g(H)\Lambda_2+t\Lambda_2 g(\Lambda_2 H\Lambda_2)\Lambda_2 = t(\Lambda_2 g(\Lambda_2 H\Lambda_2)\Lambda_2-\Lambda_2 g(H)\Lambda_2)+\Lambda_2 g(H)\Lambda_2\,.$$
	Clearly $\mathbb{\Lambda}_1 W(t)$ continuously interpolates between the two operators of interest.  So it is enough to show that this interpolation is within $\Fredholms{\HH}$. Note that $g(H) = P$ and define $Q:=g(\Lambda_2 H\Lambda_2) $. We observe that under our spectral gap assumption $P$ is a projection whereas $Q$ is \emph{not}. Then a computation shows that for all $t\in[0,1]$,
	\begin{align*}
		A(t)^2-A(t) &=  t^2( \Lambda_2Q\Lambda_2 (Q-P)\Lambda_2 + \Lambda_2P \Lambda_2 (P-Q)\Lambda_2)+t ( \Lambda_2(Q-P) \Lambda_2 P+\\ 
		&+  \Lambda_2 P \Lambda_2 (Q - P)\Lambda_2 + \Lambda_2(P-Q)\Lambda_2)+\Lambda_2[P\Lambda_2 , \Lambda_2 P]\Lambda_2 
	\end{align*}
	Now, since $\Lambda_2(Q-P)\Lambda_2$ is $\LOC{2}$ by \cref{eq:remark for final BEC homotopy} and $Q,P$ are local all the terms with a factor of $\Lambda_2(Q-P)\Lambda_2$ are $\LOC{2}$.
	
	Then as $[P\Lambda_2 , \Lambda_2 P]$ can be written as the sum of two commutators of $\Lambda_2$ with local operators it is $\LOC{2}$.  Thus, $A(t)^2-A(t)$ is $\LOC{2}$ so by \cref{eq:key lemma to show homotopies are Fredholm} we conclude that our interpolation remains Fredholm.  Furthermore, under the assumption that $H$ is TRI we have $[A(t), \Theta] =0$ so our interpolation is within $\calF_\Theta(\HH)$.    
\end{proof}

\medskip

\noindent\textbf{Acknowledgements:} JS is thankful to Alex Bols, Martin Fraas and Gian Michele Graf for useful discussions. JS acknowledges support of the Columbia mathematics department for supporting an undergraduate research project which resulted in the present report. This research is supported in part by Simons Foundation Math + X Investigator Award \#376319 (Michael I. Weinstein).

\section{Appendix}
\subsection{More proofs}
\begin{proof}[Proof of the first statement in \cref{thm:equivalence of bulk indices}]
We note that the Fermi projection $P \equiv \chi_{(-\infty, 0)}(H)$ can be written as $P= g(H)$ in the spectral gap regime.  Then using a special case of the Fedosov formula, \cref{prop:Fedosov type formula via index of pair of projections}, we may write the RHS of \cref{eq:relationship between F and Kitaev bulk} as 
\begin{equation} \label{eq:fedosov formula}
\findex\mathbb{W}_1( P\Lambda_2P) = \ii \tr\exp({-2\pi iP\Lambda_2P}) \partial_1 \exp({2\pi iP\Lambda_2P})    \,.
\end{equation}
To do so, we must demonstrate that $\partial_1 \exp({2\pi iP\Lambda_2P})$ is trace class. Since $P$ is local, $\partial_2 P$ is $\LOC{2}$ and so by the ideal property, $$ (P\Lambda_2P)^2-P\Lambda_2P=\ii P\Lambda_2P^\perp\partial_2 P $$ is $\LOC{2}$ too. Hence we may apply \cref{prop: A^2-A LOC2 then the derivative is compact} on $P\Lambda_2P$ to get that $\partial_1 \exp({2\pi iP\Lambda_2P})$ is trace class. 

Next, we show that the RHS of \cref{eq:fedosov formula} and the RHS of \cref{eq:bulk index} can be written as the same integral:	
\begin{equation}\label{eq:intermediate integral}
I := \int_{0}^{2\pi}  \tr \partial_{\alpha} ( \exp({-\ii\alpha P\Lambda_2P}) P\Lambda_1P \exp({\ii\alpha P\Lambda_2P})) \dif{\alpha}\,.
\end{equation}
Note that this formula is well-defined since the operator within the trace is trace-class:
\begin{align*}
\partial_{\alpha} ( \exp(-\ii\alpha P\Lambda_2P) P\Lambda_1P\exp(\ii\alpha P\Lambda_2P))
&= \ii \exp(-\ii\alpha P\Lambda_2P) [P\Lambda_1P, P\Lambda_2P] \exp(\ii\alpha P\Lambda_2P) \\
&= -\ii \exp(-\ii\alpha P\Lambda_2P) P[\partial_1 P,\partial_2 P] \exp(\ii\alpha P\Lambda_2P)\,.
\end{align*} 
Then expanding the integrand of \cref{eq:intermediate integral} and using cyclicality of the trace and then integrating we obtain $I$ equals the RHS of \cref{eq:bulk index}.

Conversely, as the integral in \cref{eq:intermediate integral} converges strongly and the integrand is trace class we may exchange the trace and integral and then use the fundamental theorem of calculus to obtain  
\begin{equation}
I = \tr \exp({-2\pi i P\Lambda_2P}) [P\Lambda_1P, \exp({2\pi i P\Lambda_2P})] =  \ii\tr \exp({-2\pi i P\Lambda_2P}) \partial_1 \exp({2\pi i P\Lambda_2P})
\end{equation}	
where the last equality follows from \cref{lem:aux lemma to connect two Kitaev indices}. Thus, we have shown that the RHS of \cref{eq:relationship between F and Kitaev bulk} is equal to the RHS of \cref{eq:bulk index}. Now \cite[eq-n (4.6)]{Aizenman_Graf_1998}, which is $\sigma_{\mathrm{Hall}}(H)=\frac{1}{2\pi}\findex F$, implies the first claim of \cref{thm:equivalence of bulk indices}.
\end{proof}	
\begin{prop}[A Fedosov type formula]\label{prop:Fedosov type formula via index of pair of projections}
	If $Q$ is a projection and $V$ is a unitary such that $[Q,V]$ is trace class then (using \cref{eq:projection truncation super operator}) $$\findex\mathbb{Q}V=\tr V[Q,V^\ast]\,.$$ 
\end{prop}
\begin{proof}
	This is a special case of \cite[Prop. 2.4]{ASS1994_charge_def}, where the difference of projections $Q-V^\ast Q V$ is trace class.
\end{proof}
	
\begin{lem}\label{lem:aux lemma to connect two Kitaev indices}
	If $V$ is unitary and $Q,R$ are projections such that $V = \mathbb{R}V$ and $[Q,V]$ is trace class then $$ \tr V^\ast[Q,V] = \tr V^\ast[RQR,V] \,. $$
\end{lem}
\begin{proof}
	Since $V = \mathbb{R}V$ we have $[R,V]=0$, so that (using cyclicity) \begin{align*} \tr V^\ast[RQR,V] = \tr R V^\ast R[Q,V]  = \tr (V^\ast - R^\perp) [Q,V]\,,\end{align*}
	but using cyclicity again $\tr R^\perp [Q,V] = \tr R^\perp [Q,V] R^\perp $ and since $R^\perp V = R^\perp$ this term vanishes.
\end{proof}

\subsection{Locality properties}

\begin{prop}\label{prop:main criterion for Fredholm of truncated ops}	If $A\in \Fredholms{\HH}$ and $\partial_1 A\in \Compacts{\HH}$ (the ideal of compact operators) then $$\mathbb{\Lambda}_1 A \in \calF(\HH)\,.\label{eq:fredholm condition}$$ 
\end{prop}
\begin{proof}
	Recall Atkinson's characterization of Fredholm operators: they are operators invertible modulo a compact operator \cite{Booss_Topology_and_Analysis}. Thus it suffices to construct an explicit parametrix for the operator of interest. With $B$ a parametrix for $A$ define $G := \mathbb{\Lambda}_1(B)$ and then
	\begin{align*}
		(\mathbb{\Lambda}_1 A)G-\Id = \Lambda_1(-\ii \partial_1 A+AB-\Id)\Lambda_1
	\end{align*}  
	and similarly for $G\mathbb{\Lambda}_1 A-\Id$. Since $\partial_1 A$ is compact and $B$ is a parametrix for $A$ we find that $G$ is a parametrix for $\mathbb{\Lambda}_1 A$ indeed.
\end{proof}

\begin{prop}\label{prop: A^2-A LOC2 then the derivative is compact}
	If $A$ is local such that $A^2-A$ is $\LOC{2}$, then $$\Id-\exp(-2\pi\ii A)$$ is also $\LOC{2}$. It follows that for such $A$, $\partial_1 \exp(-2\pi \ii A)$ is trace class.
	\end{prop}
\begin{proof}
	For any $n\in\ZZ$, we may re-write
		\begin{align*}
		\ee^{-2\pi\ii nA}-\Id & = \sum_{l=1}^{\infty}\frac{1}{l!}\left(-2\pi\ii nA\right)^{l}\\
		 &=  \sum_{l=1}^{\infty}\frac{1}{l!}\left(-2\pi\ii n\right)^{l}\left(A^{l}-A\right)\tag{using $\sum_{l=1}^{\infty}\frac{1}{l!}\left(-2\pi\ii n\right)^{l} = 0$}\\
		 &=  \sum_{l=2}^{\infty}\frac{1}{l!}\left(-2\pi\ii n\right)^{l}\sum_{k=0}^{l-2}A^{k}\left(A^{2}-A\right)\tag{$l=1$ term vanishes.}\,.
	\end{align*}
	
	Next, if $A$ is local, then by the algebraic closure of course $A^k$ is local as well. However, since for us $k$ will get arbitrarily large, we need to control how the estimate gets worse as $k\to\infty$.
	
	Let $\alpha\in\NN$ be sufficiently large for $ \|A_{xy}\| \leq C_\alpha(1+\|x-y\|)^{-\alpha}$ for all $x,y\in\ZZ^2$. Then \begin{align*}
		\|(A^k)_{xy}\| \leq (C_\alpha)^k\sum_{x_1,\dots,x_{k-1}} (1+\|x-x_1\|)^{-\alpha}\dots(1+\|x_{k-1}-y\|)^{-\alpha}\,.
	\end{align*}
	Now using the triangle inequality we have $(1+\|x-y\|)^{-\alpha}(1+\|y-z\|)^{-\alpha} \leq (1+\|x-z\|)^{-\alpha}$ so that 
	\begin{align*}
	\|(A^k)_{xy}\| &\leq (C_\alpha)^k(1+\|x-y\|)^{\alpha/2}\sum_{x_1,\dots,x_{k-1}} (1+\|x-x_1\|)^{-\alpha/2}\dots(1+\|x_{k-1}-y\|)^{-\alpha/2} \\ &\leq (C_\alpha)^k(1+\|x-y\|)^{\alpha/2}\sum_{x_1} (1+\|x-x_1\|)^{-\alpha/2}\dots\sum_{x_{k-1}}(1+\|x_{k-2}-x_{k-1}\|)^{-\alpha/2} \\&=  (C_\alpha)^k(\sum_{x'} (1+\|x'\|)^{-\alpha/2})^{k-1}(1+\|x-y\|)^{-\alpha/2}\\&=: (D_{\alpha/2})^k(1+\|x-y\|)^{-\alpha/2}
	\end{align*}
	
	Since $\alpha\in\NN$ was arbitrarily large, this last term is indeed finite and we get that the rate of decay in $\|x-y\|$ is fixed and the constant grows polynomially in $k$.
	
	Going back to our initial expression, we now have (using also the $\LOC{2}$ estimate for $A^2-A$) \begin{align*}
		\|(\ee^{-2\pi\ii nA}-\Id)_{xy}\| &\leq \sum_{l=2}^{\infty}\frac{1}{l!}\left(2\pi n\right)^{l}\sum_{k=0}^{l-2}\sum_{x'}(D_{\alpha})^k(1+\|x-x'\|)^{-\alpha}C_\alpha (1+\|x'-y\|)^{-\alpha}(1+|x'_2|+|y_2|)^{-\alpha} \,.
	\end{align*}
	
	Now $\sum_{l=2}^{\infty}\frac{1}{l!}\left(2\pi n\right)^{l}\sum_{k=0}^{l-2}(D_{\alpha})^k<\infty$ and the other terms are summable and yield a $\LOC{2}$ estimate using the triangle inequality again, so that $\exp(-2\pi \ii A) - \Id$ is $\LOC{2}$.
	
	This implies (by the results of \cite[Section 3.3]{Shapiro2019}) that $\partial_1 \exp(-2\pi \ii A) = \partial_1 (\exp(-2\pi \ii A) - \Id)$ is local and confined in both directions and hence trace-class.
\end{proof}

\begin{cor} \label{eq:key lemma to show homotopies are Fredholm}
	If $A(t):[0,1] \to \BLO{\HH}$ is a continuous map such that $A(0) = A_0$ and $A(1) = A_1$ and $A(t)^2 - A(t)$ is $\LOC{2}$ for all $t$, then there is a homotopy within $\calF(\HH)$ from
	$$ 		 \mathbb{W}_1 A_0  \longrightarrow \mathbb{W}_1 A_1 $$ and if $[A(t), \Theta] = 0$ then it passes within $\calF_\Theta(\HH)$, so that both $\findex$ and $\findex_2$ agree on these two operators.
\end{cor}
\begin{proof}
	Since by \cref{prop: A^2-A LOC2 then the derivative is compact} we have that $\partial_1 \exp(-2\pi\ii A(t))$ is trace-class, it is compact, so we may apply \cref{prop:main criterion for Fredholm of truncated ops} to get that for any $t$, $\mathbb{W}_1 A(t) \in \calF(\HH)$.
\end{proof}

\subsection{The smooth functional calculus}
\label{section:control_of_functional_calculus_of_difference_of_operators}

In this section we want to establish that both $\Lambda_2(g(\Lambda_2H\Lambda_2)-g(H))\Lambda_2$ and $g(\iota^\ast H \iota)-g(\hat{H})$ are $\LOC{2}$, which is used in \cref{lem:final BEC homotopy} and \cref{prop:general BC} respectively.

First we recall \cite[Theorem 10.5]{AizenmanWarzel2016} the basic estimate:
\begin{thm}\label{thm:Combes-Thomas estimate}
(The Combes-Thomas Estimate) If $A\in \BLO{\HH}$ is exponentially local and self-adjoint then there are constants $C<\infty,\mu>0$ such that \begin{align}\norm{\left(A-z\Id\right)_{xy}^{-1}}\leq C|\Im\{z\} |^{-1}\exp(-\mu|\Im\{z\} |\norm{x-y})\qquad(x,y\in\ZZ^2;z\in\CC\setminus\RR)\,.\label{eq:the Combes-Thomas estimate}\end{align}
\end{thm}

Next, we need a result about the smooth functional calculus \cite{Davies_1995} (and see references therein):
\begin{thm}[The Helffer-Sj\"ostrand formula]
	Let $f:\RR\to\CC$ be smooth and of compact support, and let $\tilde{f}:\CC\to\CC$ be a quasi-analytic extension of it (which is supported within some strip about the real axis). This implies that for all $N\in\NN$, there is some $C_N<\infty$ such that \begin{align} |(\partial_{\bar{z}} f)(z)|\leq C_N |\Im\{z\}|^{N}\qquad(z\in\CC)\,. \label{eq:control of derivative of function in the HS smooth functional calculus formula}\end{align} Then for $A\in\BLO{\HH}$ self-adjoint, we have
	
	\begin{align}
		f(A) = \frac{1}{2\pi}\int_{z\in\CC} (\partial_{\bar{z}} f)(z) (A-z\Id)^{-1}\dif{z}\label{eq:the HS formula}
	\end{align}
	
\end{thm}

One combines these two results to obtain that the smooth functional calculus on exponentially local self-adjoint operators is polynomially local \cite[Appendix A]{Elbau_Graf_2002}, a result we have been using freely.

Finally, we also get a similar statement to \cite[Lemma A3]{Elbau_Graf_2002}, which says that the difference of the smooth functional calculus of operators whose difference is $\LOC{2}$ is itself $\LOC{2}$:
\begin{prop}\label{prop:control of difference of operators under smooth function}
	If $A,B$ are two self-adjoint exponentially local operators such that $A-B$ is exponentially $\LOC{2}$, and $f:\RR\to\CC$ is a smooth function then $f(A)-f(B)$ is also $\LOC{2}$.
\end{prop}

\begin{proof}
First note that since $A,B$ are bounded, we may WLOG assume that $f$ is of compact support, and define $K:=|\supp(\partial_z \tilde{f})|$. Hence we may use \cref{eq:the HS formula} to get $$ f(A)-f(B) = \frac{1}{2\pi}\int_{z\in\CC}(\partial_{\bar{z}} f)(z) (A-z\Id)^{-1}(B-A)(B-z\Id)^{-1}\dif{z}\,. $$

For any $N\in\NN$, taking the $n,m\in\ZZ^2$ matrix elements, using the fact $f$ has compact support and \cref{eq:control of derivative of function in the HS smooth functional calculus formula} as well as \cref{eq:the Combes-Thomas estimate} we find: 

\begin{align*}
	\|(f(A)-f(B))_{n,m}\| \leq \frac{K}{2\pi}\int_{y\in\RR} C_N|y|^N\sum_{l,k} C |y|^{-1} \ee^{-\mu y \|n-l\|}\|(B-A)_{l,k}\|C |y|^{-1} \ee^{-\mu y \|k-m\|} \dif{y}\,.
\end{align*}
Now since $A-B$ is assumed $\LOC{2}$ we have some $D<\infty,\nu>0$ such that $$ \|(B-A)_{l,k}\| \leq D \exp(-\nu(\|l-k\|+|l_2|+|k_2|))\qquad(l,k\in\ZZ^2) $$ and so all together

\begin{align*}
\|(f(A)-f(B))_{n,m}\| &\leq \frac{K C_N C^2D}{2\pi}\int_{y\in\RR} |y|^{N-2}\sum_{l,k} \ee^{-\mu y \|n-l\|-\nu(\|l-k\|+|l_2|+|k_2|)-\mu y \|k-m\|}  \dif{y} \\ &\leq C'\int_{y\in\RR} |y|^{N-2} \ee^{-\mu' y (\|n-m\|+|n_2|+|m_2|)}  \dif{y}\,.
\end{align*}
Multiple integrations by parts to get rid of the $|y|^{N-2}$ factor in the integrand yield now polynomial decay at rate $N-1$.
 
\end{proof}

\begin{rem}\label{eq:remark for final BEC homotopy}

We note that the first expression want to control is re-written as \begin{align*}
	\Lambda_2(g(\Lambda_2H\Lambda_2)-g(H))\Lambda_2 &= \Lambda_2(g(\Lambda_2H\Lambda_2+\Lambda_2^\perp H\Lambda_2^\perp)-g(H))\Lambda_2
\end{align*}
and since \begin{align*}\Lambda_2H\Lambda_2+\Lambda_2^\perp H\Lambda_2^\perp-H = 2\Re\{\Lambda_2 H\Lambda_2^\perp\}=2\Re\{\ii(\partial_2 H)\Lambda_2^\perp\}  \end{align*} is $\LOC{2}$ by applying $\partial_2$ on a local operator $H$, we may apply \cref{prop:control of difference of operators under smooth function} on it to get that $\Lambda_2(g(\Lambda_2H\Lambda_2)-g(H))\Lambda_2$ is $\LOC{2}$ indeed.

The second expression admits a direct application of \cref{prop:control of difference of operators under smooth function} due to our hypothesis in \cref{def:compatibility between bulk and edge}.


\end{rem}

\subsection{More general boundary conditions}
We want to generalize \cref{eq:connection between edge Fredholm index and bulk Fredholm index} to any boundary conditions. This is achieved using the fact that
\begin{prop}\label{prop:general BC}
	If $\hat{H}$ and $H$ are compatible as in \cref{def:compatibility between bulk and edge}, then there is a homotopy within $\calF(\HH)$ from $$ \mathbb{W}_1 g(\hat{H})  \longrightarrow \mathbb{W}_1 g(\operatorname{Ad}_{\iota^\ast}H) $$ and if $[\hat{H},\Theta]=0$ and $[H,\Theta]=0$ then it passes within $\calF_\Theta(\HH)$, so that both $\findex$ and $\findex_2$ agree on these two operators.
\end{prop} 

\begin{proof}
		  
	Consider the continuous map $M:[0,1] \to \BLO{\HH}$ given by $M(t) = \exp(-2\pi \ii A(t))$, with $A(t) = t g(\operatorname{Ad}_{\iota^\ast}H) + (1-t) g(\hat{H}) $. After some algebra, we have 
	\begin{align*}
		A(t)^2-A(t) &= (t^2 g(\operatorname{Ad}_{\iota^\ast}H) - t^2 g(\hat{H}) + t g(\hat{H})  + tg(\hat{H}) - \Id)\times( g(\operatorname{Ad}_{\iota^\ast} H) - g(\hat{H})) + \\
		& \quad+ g(\hat{H})^2 - g(\hat{H})\,.
	\end{align*}
	Now, $g(\operatorname{Ad}_{\iota^\ast}H) - g(\hat{H})$ is $\LOC{2}$ by the compatibility condition \cref{def:compatibility between bulk and edge} and application of \cref{prop:control of difference of operators under smooth function}. So by the algebraic closure of $\LOC{2}$, the first line is $\LOC{2}$. Finally, we note that $g^2-g$ is only supported within the bulk spectral gap (since outside of it it takes either the value $1$ or $0$, in each case $g^2-g$ is zero). Hence $(g^2-g)(\hat{H})$ is $\LOC{2}$ as well by \cite[Lemma A3 (iii)]{Elbau_Graf_2002} (we apply it with $G:=g^2-g$). 
	
	So by \cref{eq:key lemma to show homotopies are Fredholm} we conclude the continuous interpolation between both operators remains Fredholm.  If $\hat{H}$ and $H$ are TRI then in particular we have $[A(t), \Theta] = 0$ hence the interpolation passes within $\calF_\Theta(\HH)$.  
\end{proof}

\subsection{Equivalence of $\ZZ_2$ indices}
\label{sec:equivalence of z_2 indices}
In order to see that our definition of $\mathcal{N}$ is equivalent to the Fu-Kane-Mele \cite{Kane_Mele_2005_PhysRevLett.95.146802,Fu_Kane_2007} invariant, to the Schulz-Baldes \cite{Schulz-Baldes_2015_Z2} invariant, to the Katsura-Koma \cite{Katsura_Koma_2016} invariant, as well as to the Graf-Porta \cite{Graf2013}, as well as prove \cref{eq:equivalence of bulk z2 indices}, we construct a model with fermionic time reversal symmetry and then show that our definition of $\mathcal{N}$ agrees with the the Schulz-Baldes \cite{Schulz-Baldes_2015_Z2} invariant (which has been related to the other invariants already) on both the trivial and non-trivial classes.

Let $H$ be some Hamiltonian (not necessarily such that $[H,\Theta]=0$). On a double Hilbert space $\hat{\HH}:=\HH\oplus\HH$ define $\tilde{H} := H \oplus \Theta H \Theta^*$ and 
	\begin{align*}
		\tilde{\Theta} : = \begin{bmatrix}
			0 & \Theta \\
			\Theta & 0 
		\end{bmatrix}
	\end{align*}
	A calculation shows that $\tilde{\Theta}^2 = -\Id$ and $[\tilde{H}, \tilde{\Theta}] =0$. So we see that $\tilde{H}$, $\tilde{\Theta}$, $\HH\oplus \HH$ defines a model with fermionic time reversal symmetry. 
	From $\tilde{H}$ we may naturally define the Fermi projection  
	\begin{align*}
		\tilde{P} := \chi_{(-\infty, 0)}(\tilde{H}) = \begin{bmatrix}
			\chi_{(-\infty, 0)}(H) & 0 \\
			0 & \chi_{(-\infty, 0)}(\Theta H \Theta^*)
		\end{bmatrix}
	\end{align*}
	so that $\Theta$ being anti-unitary yields, via $R(z) = \Theta^\ast(\Theta H \Theta^\ast -\bar{z}\Id)^{-1}\Theta$ and Stone's formula for $\chi_{(-\infty, 0)}$,
	\begin{align*}
		\tilde{P} = \begin{bmatrix}
			P & 0 \\
			0 & \Theta P \Theta^*
		\end{bmatrix}\,.
	\end{align*}	
	We note also that $X_j$ on $\HH$  extends naturally to $\HH\oplus \HH$ as $\tilde{X_j} = X_j \oplus X_j$ and that since $[\Theta, X_j] = 0$ we have $[\tilde{\Theta}, \tilde{X_j}] = 0$.  We have a similar extension of $\Lambda_j$ on $\HH$ to $\tilde{\Lambda}_j$ on $\HH\oplus \HH$ as well.  
	      
Having constructed the model with fermionic time reversal symmetry the equivalence of our invariant to the known invariants is shown by the following proposition.
\begin{prop}\label{prop:Z2 index equivalence}
	For the model for fermionic TRI symmetry defined previously in this section we have 
	\begin{align}
		\findex_2\tilde{\mathbb{\Lambda}}_1\exp(-2\pi \ii \tilde{P}\tilde{\Lambda}_2\tilde{P}) = \findex_2\tilde{\mathbb{P}}\tilde{U}\label{eq:equivalence of our two bulk z2 indices}
	\end{align}	
	where $\tilde{U} = \exp(\ii \arg(\tilde{X}_1+\ii \tilde{X}_2))$ is the unitary operator implementing the gauge transformation associated with a flux insertion at the origin.   
\end{prop}

\begin{proof}
	Using the anti-unitary property of $\Theta$ and commutation relations discussed previously after some algebra we have
	\begin{align*}
		\tilde{\mathbb{\Lambda}}_1\exp(-2\pi \ii \tilde{P}\tilde{\Lambda}_2\tilde{P}) = \begin{bmatrix}
			\mathbb{W}_1( P \Lambda_2 P) & 0\\
			0 & \Theta\mathbb{W}_1( -P \Lambda_2 P)\Theta^*
		\end{bmatrix}
	\end{align*}
	We stress the different sign in second block. Then using the elementary fact that $\dim \ker(A\oplus B) = \dim\ker A+\dim \ker B$, the fact that $\Theta$ is a bijection (so it doesn't change dimensions of kernels), and the fact that $n+m\mod2=n-m\mod2$, we get 
	\begin{align*}
		\findex_2\tilde{\mathbb{\Lambda}}_1\exp(-2\pi \ii \tilde{P}\tilde{\Lambda}_2\tilde{P}) &= \dim\ker\mathbb{W}_1( P \Lambda_2 P)+\dim\ker\Theta\mathbb{W}_1( -P \Lambda_2 P)\Theta^\ast \mod2\\
		 &= \dim\ker\mathbb{W}_1( P \Lambda_2 P)+\dim\ker\mathbb{W}_1( -P \Lambda_2 P)\mod2\\
		 &= \dim\ker\mathbb{W}_1( P \Lambda_2 P)-\dim\ker\mathbb{W}_1( -P \Lambda_2 P)\mod2\\
		 &\equiv \findex\mathbb{W}_1( P \Lambda_2 P)\mod 2
	\end{align*}
	where one notes that the last line is the \emph{usual} Fredholm index. Then by the first statement of \cref{thm:equivalence of bulk indices} we may write 
	\begin{align*}
		\findex_2\tilde{\mathbb{\Lambda}}_1\exp(-2\pi \ii \tilde{P}\tilde{\Lambda}_2\tilde{P}) =  \findex\mathbb{P}U \mod 2\,.
	\end{align*} 
	Again using the anti-unitary property of $\Theta$, commutation relations, and the observation that $U$ satisfies \cref{eq:the theta-odd constraint} we \emph{also} have 
	\begin{align*}
		\tilde{\mathbb{P}}\tilde{U} = \begin{bmatrix}
			\mathbb{P}U & 0\\
			0 & \Theta\mathbb{P}U^\ast\Theta^*
		\end{bmatrix}
	\end{align*}
	so that following the same procedure, we obtain \begin{align*}
		\findex_2\tilde{\mathbb{P}}\tilde{U} &= \findex\mathbb{P}U\mod2
	\end{align*}
	and so \cref{eq:equivalence of our two bulk z2 indices}. 
\end{proof}

In conclusion, since all the pre-existing $\ZZ_2$ indices are known to also admit such a direct sum decomposition relating them to the Chern number, our index agrees with them.

\subsection{$\Theta$-odd Fredholm theory}\label{subsec:theta-odd Fredholm theory}
For convenience of the reader, we include here a repetition of some of the $\ZZ_2$ Fredholm theory phrased via the $\Theta$-odd constraint \cref{eq:the theta-odd constraint}. These proofs first appeared in \cite{Atiyah1969} under the guise of \emph{skew-adjoint} Fredholm theory and then new proofs were presented in \cite{Schulz-Baldes_2015_Z2} for \emph{odd-symmetric} operators. In what follows, $\mathcal{B}_\Theta(\HH),\mathcal{K}_\Theta(\HH)$ are the bounded linear operators and respectively compact operators obeying \cref{eq:the theta-odd constraint}.

\begin{thm}\label{thm:continuity of index2}
	$\findex_2$ is stable under norm continuous perturbations obeying \cref{eq:the theta-odd constraint}.     	
\end{thm}

\begin{proof}
	Let $T\in \calF_\Theta(\HH)$.  Then we may make the decompositions $\HH = \ker(T)^\perp\oplus \ker T$ and $\HH = \im T\oplus \coker T$ by $\im T\in\Closed{\HH}$.  With respect to this decomposition we may write  
	\begin{align*}
		T = \begin{bmatrix}
			T_{11} & 0 \\
			0 & 0 
		\end{bmatrix}\,,\quad \Theta = \begin{bmatrix}
			\Theta_{11} & 0 \\
			0 & \Theta_{22}
		\end{bmatrix}
	\end{align*} 
	where $T_{11}:\ker(T)^\perp \to \im T$ is an isomorphism. We note that $\Theta$ is diagonal in this decomposition due to $T=-\Theta T^\ast \Theta$. Furthermore, $\Theta_{11},\Theta_{22}$ are both anti-unitary operators such that $\Theta_{11}^2 = \Theta_{22}^2 = -\Id$ due to the corresponding properties of $\Theta$. 
	
	Now let $S\in \mathcal{B}_\Theta(\HH)$ be given with $S_{ij}$ for $1\leq i,j\leq 2$ denoting the blocks of $S$ with respect to the stated decomposition of $\HH$. For $\norm{S}$ is sufficiently small we have $T_{11}+S_{11}$ invertible. Define $A := -S_{21}(T_{11}+S_{11})^{-1}S_{12}+S_{22}$. Since $A$ is a linear map from $\ker T \to \coker T$ and $T$ is Fredholm (so it has finite kernel), by the rank nullity theorem we have $\dim\ker T=\dim\im A+\dim\ker A$. Performing an LDU-decomposition for $T+S$, we find  $$T+S = I_1 ((T_{11}+S_{11})\oplus A) I_2$$ where $I_1,I_2$ are invertible. Since $T_{11}+S_{11}$ is invertible, we conclude $\ker(T+S)=\ker A$ so that
	\begin{align*}
		\dim \ker(T+S) + \dim \im A= \dim \ker T\,.
	\end{align*} 
	Then as $A$ is $\Theta_{22}$-odd and it is a map between finite $\ker T\to\coker T\cong\ker T$ (as $\findex T=0$) we may use \cref{lem: even dimensional image} to conclude that $\dim \ker(T+S) = \dim \ker T \mod 2$, i.e., $\findex_2(T+S) = \findex_2 T$.  
	\end{proof}
	
\begin{lem}\label{lem: even dimensional image}
	For $A\in \mathcal{B}_{\Theta}(\mathcal{V})$, if $\mathcal{V}$ is finite dimensional then $\dim\im A$ is even.
\end{lem}

\begin{proof}
	Let $\psi \in \im A$.  Consider $A\Theta \psi \in \im A$.  Noting $A$ is $\Theta$-odd and $\Theta$ is anti-unitary a calculation shows $\langle A\Theta \psi, \psi \rangle =0$.  As $\im A \cong \ker(A)^\perp$ and $\Theta$ is invertible we conclude that each element of $\im A$ is paired with a distinct orthogonal element.  Hence, $\dim \im A$ is even.       
\end{proof}

\begin{thm}\label{thm:compact perturbation of ind2}
	$\findex_2$ is stable under perturbations in $\mathcal{K}_\Theta(\HH)$.
\end{thm}
\begin{proof}
	The fact that $\calF_\Theta(\HH)+\mathcal{K}_\Theta(\HH)\subseteq \calF_\Theta(\HH)$ follows from Atkinson's theorem \cite{Booss_Topology_and_Analysis}. Then given $T\in \calF_\Theta(\HH)$ and $K\in \mathcal{K}_\Theta(\HH)$ consider the norm continuous map $[0,1] \ni t \mapsto T+tK$.  This map passes within $\calF_\Theta(\HH)$ by the first statement.  Hence, by \cref{thm:continuity of index2} the result follows.   
\end{proof}

\begingroup
\let\itshape\upshape
\printbibliography
\endgroup
\end{document}